\def\fullversion{1}
\def\draft{0}

\ifnum\fullversion=1
  \newcommand{\apref}[1]{Appendix~\ref{#1}}
\else
  \newcommand{\apref}[1]{the full version of the paper}
\fi

\ifnum\draft=1
  \documentclass[a4paper,envcountsame,runningheads,draft]{llncs}
\else
  \documentclass[a4paper,envcountsame,runningheads]{llncs}
\fi

\usepackage[utf8]{inputenc}
\usepackage[T1]{fontenc}
\usepackage[english]{babel}
\usepackage{amsmath,amssymb,amsfonts,latexsym}
\usepackage{newalg}
\usepackage{pgfplots}
\usepackage{graphicx,tikz}
\usetikzlibrary{decorations.shapes, calc}
\usepackage{newalg}

\spnewtheorem*{remark*}{Remark}{\bfseries}{\itshape}

\newcommand{\OPT}{\ensuremath{\mathrm{OPT}}}
\newcommand{\VAL}{\ensuremath{\mathrm{VAL}}}
\newcommand{\PERM}{\ensuremath{\mathrm{PERM}}}
\newcommand{\AVG}{\ensuremath{\mathrm{AVG}}}

\newcommand{\child}[1]{\mathsf{Ch}(#1)}
\newcommand{\BEAS}{\begin{eqnarray*}}
\newcommand{\EEAS}{\end{eqnarray*}}
\newcommand{\BEA}{\begin{eqnarray}}
\newcommand{\EEA}{\end{eqnarray}}
\newcommand{\BEQ}{\begin{equation}}
\newcommand{\EEQ}{\end{equation}}
\newcommand{\BIT}{\begin{itemize}}
\newcommand{\EIT}{\end{itemize}}
\newcommand{\BNUM}{\begin{enumerate}}
\newcommand{\ENUM}{\end{enumerate}}

\let\doendproof\endproof
\renewcommand\endproof{~\hfill\qed\doendproof}

\title{Sublinear-Time Algorithms for Monomer-Dimer Systems on Bounded Degree Graphs
\texorpdfstring{\ifnum\fullversion=0
    \protect\footnote{Full version available at {\url{http://arxiv.org/abs/1208.3629}}}
\fi}{}
}
\titlerunning{Sublinear-Time Algorithms for Monomer-Dimer Systems}
\author{Marc Lelarge \and Hang Zhou}
\institute{INRIA, \'Ecole Normale Sup\'erieure, France\\\email{\{marc.lelarge,hang.zhou\}@ens.fr}}

\date{}

\usepackage[pdfpagemode=UseNone]{hyperref}
\hypersetup{
     colorlinks=false,
     pdfborder={0 0 0}
}

\begin{document}


\maketitle

\begin{abstract}
For a graph $G$, let $Z(G,\lambda)$ be the partition function of the monomer-dimer system defined by $\sum_k m_k(G)\lambda^k$, where $m_k(G)$ is the number of matchings of size $k$ in $G$. We consider graphs of bounded degree and develop a sublinear-time algorithm for estimating $\log Z(G,\lambda)$ at an arbitrary value $\lambda>0$ within additive error $\epsilon n$ with high probability. The query complexity of our algorithm does not depend on the size of $G$ and is polynomial in $1/\epsilon$, and we also provide a lower bound quadratic in $1/\epsilon$ for this problem. This is the first analysis of a sublinear-time approximation algorithm for a $\# P$-complete problem. Our approach is based on the correlation decay of the Gibbs distribution associated with $Z(G,\lambda)$. We show that our algorithm approximates the probability for a vertex to be covered by a matching, sampled according to this Gibbs distribution, in a near-optimal sublinear time. We extend our results to approximate the average size and the entropy of such a matching within an additive error with high probability, where again the query complexity is polynomial in $1/\epsilon$ and the lower bound is quadratic in $1/\epsilon$. Our algorithms are simple to implement and of practical use when dealing with massive datasets. Our results extend to other systems where the correlation decay is known to hold as for the independent set problem up to the critical activity.
\end{abstract}

\section{Introduction}

The area of sublinear-time algorithms is an emerging area of computer
science which has its root in the study of massive data
sets \cite{Czumaj06sublinear-timealgorithms,MR2275720}.
Internet, social networks or communication networks are
typical examples of graphs with potentially millions of vertices
representing agents, and edges representing possible interactions among
those agents. In this paper, we present sublinear-time algorithms for
graph problems. We are concerned more with problems of counting and
statistical inference and less with optimization. For example, in a mobile call
graphs, phone calls can be represented as a matching of the graph
where each edge has an activity associated to the intensity of the
interactions between the pair of users. Given such a graphs, with
local activities on edges, we
would like to answer questions like: what is the size of a typical
matching? for a given user what is the probability of being matched?
As another example, models of statistical physics have been proposed
to model social interactions. In particular, spin systems are a general framework
for modeling nearest-neighbor interactions on graphs. In this setting,
the activity associated to each edge allows to model a perturbed
best-response dynamics \cite{blu93}. Again in this setting, it is
interesting to compute estimations for the number of
agents playing a given strategy or the probability for an agent
in the graph to play a given strategy at equilibrium.

There are now quite a few results on sublinear-time approximation algorithms
for graph optimization problems: minimum spanning tree weight~\cite{chazelle2005approximating},
minimum set cover~\cite{nguyen2008constant},
maximum matching~\cite{nguyen2008constant,yoshida2009improved} and
minimum vertex cover~\cite{nguyen2008constant,onak2012near,parron07}.
There are also a couple of works on sublinear-time algorithms for statistical and counting problems, e.g., approximating the average degree of a graph~\cite{feige2006sums,goldreich2008approximating} and approximating the number of occurrences of a certain structure (such as a star) in a graph~\cite{gonen2011counting}.
Our focus in this paper is on the algorithmic problems arising in statistical physics and classical
combinatorics \cite{wel93}.
We now present the monomer-dimer problem which will be the main focus
of our paper.

Let $G=(V,E)$ be an undirected graph with $|V|=n$ vertices and $|E|=m$
edges, where we allow $G$ to contain parallel edges and self-loops.
We denote by $N(G,v)$ the set of neighbors of $v$ in $G$. We consider
bounded degree graphs with $\max_v |N(G,v)|\leq \Delta$.
In a monomer-dimer system, the vertices are covered by non-overlapping
arrangement of monomers (molecules occupying one vertex of $G$) and
dimers (molecules occupying two adjacent vertices of $G$)
\cite{hl72}.
It is convenient to identify monomer-dimer arrangements with matchings;
a matching in $G$ is a subset $M\subset E$ such that no two edges in
$M$ share an endpoint.
Thus, a matching of cardinality $|M|=k$ corresponds exactly to
a monomer-dimer arrangement with $k$ dimers and $n-2k$ monomers.
Let $\mathbb{M}$ be the set of matchings of $G$. To each matching $M$,
a weight $\lambda^{|M|}$ is assigned, where $\lambda> 0$ is called
the activity. The partition function of the system is defined by
$Z(G,\lambda) = \sum_{M\in \mathbb{M}} \lambda^{|M|}$, and the Gibbs
distribution on the space $\mathbb{M}$ is defined by $\pi_{G,\lambda}
(M) =\frac{\lambda^{|M|}}{Z(G,\lambda)}$. The function $Z(G,\lambda)$ is also of combinatorial interest and called the {\em matching
polynomial} in this context~\cite{lp09}. For example, $Z(G,1)$ enumerates all matchings
in $G$.
From an algorithmic viewpoint, no feasible method is known for computing $Z(G,\lambda)$ exactly for general
monomer-dimes system; indeed, for any fixed value of $\lambda >0$, the
problem of computing $Z(G,\lambda)$ exactly
in a graph of bounded degree $\Delta$ is
complete for the class \#P of enumeration problems, when $\Delta\geq 5$ (see~\cite{vadhan2002complexity}).
The focus on computing $Z(G,\lambda)$ shifted to finding approximate solutions in polynomial time.
For example, the Markov Chain Monte Carlo (MCMC) method yields a provably efficient algorithm for finding an approximate solution.
Based on the equivalence between the counting problem (computing $Z(G,\lambda)$) and the sampling problem (according to $\pi_{G,\lambda}$) \cite{jvv86}, this approach focuses on rapidly mixing Markov chains to obtain appropriate random samples.
A Fully Polynomial-time Randomized Approximation Scheme (FPRAS) for computing the total number of matchings based on MCMC was provided by Jerrum and Sinclair \cite{jerbook,sinbook}.

Another related problem in the monomer-dimer system is the average size of a matching sampled
according to $\pi_{G,\lambda}$, defined by $E(G,\lambda)= \sum_{M\in \mathbb{M}}|M|\;\pi_{G,\lambda}(M)$.
Sinclair and Srivastava recently proved in~\cite{sinclair2013lee} that for any fixed value of $\lambda>0$, the problem of computing $E(G,\lambda)$ exactly in a bounded degree graph (allowing parallel edges) is \#P-hard, for any maximum degree $\Delta\geq5$.
Thus again we are interested in finding approximate solutions to this problem.

In order to study sublinear-time approximation algorithms for these problems, we use
the approach based on the concept of correlation decay
originating in statistical physics \cite{mmbook} and which has been
used to get a deterministic approximation scheme for counting
matchings in polynomial time~\cite{bayati2007simple}.
It follows already from \cite{hl72} that the marginals of the probability distribution
$\pi_{G,\lambda}$ are local in nature: the local structure of the
graph around a vertex $v$ allows to compute an approximation of the
corresponding marginal. In the computer science literature, this
property follows from the so-called \emph{correlation decay property}.
Our algorithm is then simple to understand: we
need only to sample a fixed number of vertices, approximate the marginals
associated to these vertices locally and then from these values output
an estimate for the desired quantity.
The correlation decay property also holds for other systems such as the independent set problem~\cite{weitz2006counting}, the coloring problem~\cite{gamarnik2007correlation}, and  the two-state spin system~\cite{liluyin12,DBLP:conf/soda/LiLY13,sinclair2012approximation}.
In \apref{sec:other}, we extend our technique to the independent set problem. We believe that similar extensions can be done for other systems as soon as the correlation decay property holds.

A graph $G$ is represented by two kinds of oracles $\mathcal{D}$ and $\mathcal{N}$
such that $\mathcal{D}(v)$ returns the degree of $v\in V$ and $\mathcal{N}(v,i)$ returns the $i^{\rm th}$ (with $1\leq i\leq
\mathcal{D}(v)$) neighbor of $v\in V$. The efficiency of an algorithm is
measured by its query complexity, i.e. the total number of accesses to
$\mathcal{D}$ and $\mathcal{N}$.
Let $\VAL$ denote a real value associated with the graph.
We say that $\widehat{\VAL}$ is an \emph{$\epsilon$-approximation} of $\VAL$ if $\widehat{\VAL}-\epsilon \leq \VAL \leq \widehat{\VAL}+\epsilon$, where $\epsilon>0$ is specified as an input parameter.
An algorithm is called an \emph{$\epsilon$-approximation algorithm} for
$\VAL$ if for any graph $G$, it computes an $\epsilon$-approximation of $\VAL$ with high probability (e.g., at least $\frac{2}{3}$).
In our model, we consider the case of
constant maximum degree $\Delta$ as $\epsilon$ tends to zero, i.e., we always first take the limit as $\epsilon\to0$ and then the limit $\Delta\to\infty$.

Our main contribution (Theorem~\ref{partition}) is an $\epsilon n$-approximation algorithm for $\log Z(G,\lambda)$ in a graph $G$ of bounded degree $\Delta$.
The query complexity of the algorithm is $\tilde{O}\left((1/\epsilon)^{\tilde{O}(\sqrt{\Delta})}\right)$, which does not depend on the size of the graph.
From the relation between the partition function and the matching statistics, we then obtain $\epsilon n$-approximation algorithms for the average size of a matching and the entropy of $\pi_{G,\lambda}$ with the same query complexity as before.
We also provide the $\Omega(1/\epsilon^2)$ query lower bound for $\epsilon n$-approximation algorithms for $\log Z(G,\lambda)$ and the other two problems.

The main tool of the above algorithms is the approximation of the marginal $p_{G,\lambda}(v)$, which is the probability that the vertex $v$ is not covered by a matching under the Gibbs distribution. We estimate $p_{G,\lambda}(v)$ for an arbitrary
vertex $v\in V$ within an error of $\epsilon>0$ with near-optimal
query complexity $\tilde{O}\left((1/\epsilon)^{\tilde{O}(\sqrt{\Delta})}\right)$.

The rest of the paper is organized as follows.
In Section~\ref{sec:loc}, we prove our first main result concerning local computations for matchings.
Based on this result, we construct an $\epsilon n$-approximation algorithm for the partition function $Z(G,\lambda)$ in Section~\ref{sec:partition function} and $\epsilon n$-approximation algorithms for the average size of a matching and the entropy of $\pi_{G,\lambda}$ in Section~\ref{sec:matching statistics}.
We also provide query lower bounds in these two sections.
In Section~\ref{sec:applications}, we give some applications of our technique for approximating the permanent of constant degree expander graphs and the size of a maximum matching (in this last case, our algorithm is outperformed by \cite{yoshida2009improved}).
In \apref{sec:tests}, we also show the efficiency of our algorithms by testing on large real-world networks.

\section{Local Computations for Matchings}\label{sec:loc}

Recall that we defined for all $\lambda>0$, the Gibbs distribution on
matchings of a graph $G$ by:
\BEAS
\forall M\in \mathbb{M}, \quad \pi_{G,\lambda}(M) =
\frac{\lambda^{|M|}}{Z(G,\lambda)}\mbox{ where } Z(G,\lambda) = \sum_{M\in \mathbb{M}} \lambda^{|M|}.
\EEAS
The focus in this section is on the approximation of the probability that a vertex $v\in V$ is not coverd by a matching:
\[p_{G,\lambda}(v)  := \sum_{M \not\ni v}\pi_{G,\lambda} (M),\]
where $M \not\ni v$ is a matching not covering $v$.

First notice that
\BEA
\label{eq:pi}p_{G,\lambda}(v) = \frac{Z(G\backslash \{v\},\lambda)}{Z(G,\lambda)},
\EEA
where $G\backslash \{v\}$ is the graph obtained from $G$ by removing
the vertex $v$ and all incident edges.
Then we have
\BEAS
Z(G,\lambda) = Z(G\backslash \{v\},\lambda)+\lambda \sum_{u\in
  N(G,v)}Z(G\backslash \{u,v\},\lambda),
\EEAS
so that dividing by $Z(G\backslash \{v\},\lambda)$, we get
\BEA
\label{eq:recpi}p_{G,\lambda}(v) =\frac{1}{1+\lambda \sum_{u\in N(G,v)} p_{G\backslash \{v\},\lambda}(u)}.
\EEA
This recursive expression for $p_{G,\lambda}(v)$ is well-known
and allows to compute the marginal $p_{G,\lambda}(v)$ exactly for each
$v\in V$. We follow the approach of Godsil \cite{god81}.
First, we recall the notion of \textit{path-tree} associated with a rooted graph:
if $G$ is any rooted graph with root $v_0$, we define
its path-tree $T_G(v_0)$ as the rooted tree whose vertex-set consists of
all \emph{finite simple paths} starting at the root $v_0$; whose edges are the pairs $\{P,P'\}$ of the
  form $P=v_0\ldots v_k$, $P'=v_0\ldots v_k v_{k+1} (k\geq 0)$; and whose
  root is the single-vertex path $v_0$.
By a \textit{finite simple path}, we mean here a finite sequence of
distinct vertices $v_0\ldots v_k$ ($k\geq 0$) such that  $v_i v_{i+1}\in E$
for all $0\leq i < k$. Note that the notion of path-tree is similar to
the more standard notion of computation tree, the main difference
being that for any finite graph $G$, the path-tree is always finite
(although its size might be much larger than the size of the original
graph $G$).

For every node $u$ in the path-tree $T_G(v)$, define $\child{u}$ to be the set of children of $u$ in $T_G(v)$.
The recursion (\ref{eq:recpi}) easily implies
$p_{G,\lambda}(v) = p_{T_G(v),\lambda}(v)$ and
$p_{T_G(v),\lambda}(v)=x_v(v)$, where the vector ${\bf x}(v)=(x_u(v),\:u\in T_G(v))$
solves the recursion:
\BEA
\label{eq:recx}\forall u\in T_G(v),\quad  x_u(v)=\frac {1}{1+\lambda \sum_{w\in \child{u}} x_w(v)}
\EEA
(by convention a sum over the empty set is zero).

In order to approximate $p_{G,\lambda}(v)$, we will show that it suffices to solve
the recursion (\ref{eq:recx}) restricted to a truncated path-tree of $T_G(v)$.
For any $h\geq 1$, let $T^h_G(v)$ be the path-tree truncated at depth
$h$ and let ${\bf x}^h(v) = (x_u^h(v),\:u\in T^h_G(v))$ be the solution of the recursion (\ref{eq:recx})
when the path-tree is replaced by the truncated version
$T^h_G(v)$. Clearly $x_v^h(v)= p_{G,\lambda}(v)$ for any $h\geq n$ and
the following lemma gives a quantitative estimate on how large $h$
needs to be in order to get an $\epsilon$-approximation of $p_{G,\lambda}(v)$.

\begin{lemma}\label{lem:bayati}
There exists $\overline{h}(\epsilon,\Delta)$ such that $|\log x_v^{h}(v) -\log p_{G,\lambda}(v)|\leq\epsilon$ for any $h\geq \overline{h}(\epsilon,\Delta)$. Moreover $\overline{h}(\epsilon,\Delta)=\tilde{O}\left(\sqrt{\Delta}\log(1/\epsilon)\right)$ and satisfies
\[\displaystyle\lim_{\Delta\to \infty}\frac{1}{\sqrt{\Delta}} \lim_{\epsilon\to 0}
\frac{\overline{h}(\epsilon,\Delta)}{\log (1/\epsilon)} = \sqrt{\lambda}.\]
\end{lemma}
\begin{proof}
Theorem 3.2 in
\cite{bayati2007simple} proves that:
\begin{equation}
\label{eq:correlation decay}
|\log x_v^h(v) - \log p_{G,\lambda}(v)|\leq \left(1-\frac{2}{\sqrt{1+\lambda\Delta}+1}\right)^{h/2}\log(1+\lambda\Delta).
\end{equation}
The lemma then follows directly by taking $\overline{h}(\epsilon,\Delta)$ to be the $h$ such that the right-hand side equals $\epsilon$.
\end{proof}

We now present the algorithmic implication of Lemma
\ref{lem:bayati}. We start with a simple remark. The exact value for
$\overline{h}(\epsilon,\Delta)$ follows from the proof of the lemma,
however this value will not be required in what follows as shown by
the following argument: the fact that $(z_1,\dots,z_{\Delta})\mapsto
\left( 1+\lambda \sum_{i=1}^\Delta z_i\right)^{-1}$ is strictly
decreasing in each positive variable $z_i$ implies (by a simple
induction) that for any $k\geq 0$, we have
\begin{equation}
\label{alternate}
x^{2k+1}_v(v) \leq x^{2k+3}_v(v) \leq p_{G,\lambda}(v)\leq x^{2k+2}_v(v)\leq x^{2k}_v(v).
\end{equation}
Consider an algorithm that computes $x_v^h(v)$
for increasing values of $h$ and stops at the first time two
consecutive outputs are such that $|\log x_v^{h+1}(v)-\log x_v^{h}(v)|\leq \epsilon$.
By Lemma \ref{lem:bayati}, it takes at most
$\overline{h}(\epsilon,\Delta)$ iterations and the last output will be an
$\epsilon$-approximation of $\log p_{G,\lambda}(v)$.

The algorithm \textsc{Approx-Marginal}($\lambda,\epsilon,v$) provides an estimate of  $p_{G,\lambda}(v)$, based on the Depth-First-Search (DFS) on the truncated path-tree rooted at $v$.
In the algorithm \textsc{DFS}$(\lambda,h,s,\ell)$,
integer $h$ is the truncated level of the path tree $T_G(v)$;
$s\in V$ is the current node in the graph $G$ visiting by the DFS;
and $path$ maintains an array of nodes in $G$ which form the path from $v$ to $s$ during the DFS.
This path also corresponds to a node in the path-tree $T^h_G(v)$ and let $\ell$ be the length of $path$.
The algorithm \textsc{DFS}$(\lambda,h,s,\ell)$ computes recursively the marginal probability of $path$ in $T^h_G(v)$.
Recall that $\mathcal{D}(v)$ returns the degree of $v\in V$ and $\mathcal{N}(v,i)$ returns the $i^{\rm th}$ (with $1\leq i\leq \mathcal{D}(v)$) neighbor of $v\in V$.\\

\begin{algorithm}{Approx-Marginal}{\lambda,\epsilon,v}
x[1]\=\CALL{DFS}(\lambda,1,1,v)\\
x[2]\=\CALL{DFS}(\lambda,1,2,v)\\
h\=2\\
\begin{WHILE}{|\log x[h]-\log x[h-1]|>\epsilon/e}
h\=h+1\\
x[h]\=\CALL{DFS}(\lambda,h,v,1)
\end{WHILE}\\
\RETURN x[h]
\end{algorithm}
%
\begin{algorithm}{DFS}{\lambda,h,s,\ell}
\begin{IF}
{\ell=h}{\RETURN 1}
\end{IF}\\
A\=1,\ path[\ell]\=s\\
\begin{FOR}{i\=1 \TO \mathcal{D}(s)}
t\=\mathcal{N}(s,i)\\
\begin{IF}{\forall j\in [1,\ell], \; t\neq path[j]}
A\=A+\CALL{DFS}(\lambda,h,t,\ell+1)
\end{IF}
\end{FOR}\\
\RETURN 1/(\lambda A)
\end{algorithm}

\begin{proposition}
\label{marginal}
The algorithm \textsc{Approx-Marginal}($\lambda,\epsilon,v$) gives an estimate $\widehat{p}$ of $p_{G,\lambda}(v)$, such that $|\widehat{p}-p_{G,\lambda}(v)|$ and $|\log \widehat{p}-\log p_{G,\lambda}(v)|$ are both smaller than $\epsilon$.
Its query complexity is $\overline{\mathcal{Q}}(\epsilon,\Delta)=\tilde{O}\left((1/\epsilon)^{\tilde{O}(\sqrt{\Delta})}\right)$. In addition, $Q=\overline{\mathcal{Q}}(\epsilon,\Delta)$ satisfies
\BEA\label{eq:limT}
\lim_{\Delta\to \infty}\frac{1}{\sqrt{\Delta}\log\Delta} \lim_{\epsilon\to 0}
\frac{\log Q}{\log (1/\epsilon)} = \sqrt{\lambda}.
\EEA
\end{proposition}

\begin{proof}
Let $h$ be the final truncated level of the path-tree $T_G(v)$ in the algorithm.
We have $\widehat{p}=x_v^h(v)$ and $|\log x_v^h(v)-\log x_v^{h-1}(v)|<\epsilon/e$.
Thus $|\log \widehat{p}-\log p_{G,\lambda}(v)|< \epsilon/e$ by Inequality~\eqref{alternate}.
Since $ p_{G,\lambda}(v)$ and $\widehat{p}$ are at most 1, we then have $|\widehat{p}-p_{G,\lambda}(v)|< \epsilon$.
The number of nodes visited by the algorithm is $O\left(\Delta^{h}\right)$, so the number of queries is also $O\left(\Delta^{h}\right)$.
The proposition follows by applying the upper bound $\overline{h}(\epsilon,\Delta)$ on $h$ from Lemma \ref{lem:bayati}.
\end{proof}

\begin{remark*}
In Section~\ref{sec:partition function}, we need to estimate the marginal probability at the node $v$ in the graph $G_{v}=\{u\in V\mid u \succeq v\}$ instead of the graph $G$, where $\succ$ is some total order over $V$.
To achieve this, we only need to add an additional constraint $t\succeq v$ to line~6 of the DFS algorithm.
Denote \textsc{Approx-Marginal$^*$}$({\lambda,\epsilon,v})$ to be the modified version of \textsc{Approx-Marginal}$({\lambda,\epsilon,v})$ with the underlying graph $G_{v}$.
Again, Proposition~\ref{marginal} holds for the algorithm \textsc{Approx-Marginal$^*$}$({\lambda,\epsilon,v})$ by replacing $G$ by $G_v$.
\end{remark*}

The next propostion shows that there exists some $\underline{\mathcal{Q}}(\epsilon,\Delta)$, such that $Q=\underline{\mathcal{Q}}(\epsilon,\Delta)$ satisfies Equation~\eqref{eq:limT} and that $\underline{\mathcal{Q}}(\epsilon,\Delta)$ is a query lower bound for computing an $\epsilon$-approximation of $p_{G,\lambda}(v)$. This implies that Algorithm \textsc{Approx-Marginal} is optimal when the influence of $\epsilon$ is much larger than that of $\Delta$.
The idea of the lower bound proof is to construct two instances of \emph{almost full $\Delta$-ary trees} whose marginal probabilities at the root differ by more than $\epsilon$, while any approximation algorithm using a small number of queries cannot distinguish them.
See \apref{proof_lowerboundP} for a detailed proof of this proposition.

\begin{proposition}
\label{lowerboundP}
In order to approximate the marginal $p_{G,\lambda}(v)$ within an additive error $\epsilon$, any deterministic or randomized algorithm\footnote{
In the randomized case, the algorithm is expected to provide \emph{always} an estimate with an additive error $\epsilon$, and the proposition implies a lower bound on the \emph{average} number of queries of such an algorithm.
}
requires
$\Omega\left(\underline{\mathcal{Q}}(\epsilon,\Delta)\right)$ queries
where $Q=\underline{\mathcal{Q}}(\epsilon,\Delta)$ satisfies Equation~\eqref{eq:limT}.
\end{proposition}

\begin{remark*}
As noted in the introduction, the model with $\lambda_e$ ($e\in E$) varying
across the edges is of practical interest (allowing to model
various intensities on edges).
As soon as there exists $\lambda_{\max}$ such that for all $e\in E$, we
have $\lambda_e\in [0,\lambda_{\max}]$, it is easy to extend the
results of this section to the more general model defined by (note
that $\vec{\lambda}$ is now a vector in $[0,\lambda_{\max}]^E$):
$\pi_{G,\vec{\lambda}}(M) =  \frac{\prod_{e\in M} \lambda_e}{Z(G,\vec{\lambda})}$ where, $Z(G,\vec{\lambda}) =\sum_{M\in \mathbb{M}} \prod_{e\in M}\lambda_e$.
Results in this section and Sections~\ref{sec:partition function} and~\ref{sec:matching statistics} hold provided $\lambda$ is
replaced by $\lambda_{\max}$.
\end{remark*}

\section{Approximating the Partition Function}
\label{sec:partition function}

First, we need an arbitrary total order $\succ$ over $V$.
We can achieve this by assigning a random number $a_v\in[0,1]$ to each vertex $v$ and then defining $u\succ v$ as $a_u>a_v$.
However, if there are only a small number of vertices involved in our computation, we do not need to generate random numbers for all vertices.
Using the technique in~\cite{onak2012near}, we generate a random number each time we visit a new vertex and then save its value for the later visits.
Generating a random number can be done in sublinear time and the number of vertices in our computation is at most twice the number of queries,  which will later be proved to be a constant independent of $n$.
As a result, the total time complexity for this random generation is sublinear.

Define $G_{v}=\{u\in V\mid u \succeq v\}$.
The following formula which allows us to compute the partition function from the
marginals is obtained easily from (\ref{eq:pi}):
\begin{equation}
\label{z}
\log Z(G,\lambda) = \sum_{v\in V} - \log p_{G_v,\lambda}(v).
\end{equation}

The algorithm below estimates $\log Z(G,\lambda)$.
We sample $\lceil C/\epsilon^2\rceil$ vertices uniformly at random from $V$, where $C$ is some fixed constant.
For every sampled vertex $u$, we compute an estimate of the marginal $p_{G_u,\lambda}(u)$ using the algorithm
\textsc{Approx-Marginal$^*$}$(\lambda,\epsilon/2,u)$.
We then obtain an estimate of $Z(G,\lambda)$ from the estimates of marginals at sampled vertices.\\

\begin{algorithm}{Approx-Partition-Function}{\lambda,\epsilon}
s\=\lceil C/\epsilon^2\rceil\\
U\=\text{random multi-subset of $V$ with $s$ elements}\\
\RETURN  (n/s) \cdot\left(\sum_{u\in U}-\log(\CALL{Approx-Marginal$^*$}(\lambda,\epsilon/2,u))\right)
\end{algorithm}

We recall a basic lemma which follows from Hoeffding's inequality and which will be used several
times in the paper.
\begin{lemma}
\label{sampling}
(see \cite{canetti1995lower})
Let V be a set of $n$ real numbers in $[A,B]$, where $A$ and $B$ are constant. Let $V'$ be a multi-subset of $V$ consisting of $\Theta(1/\epsilon^2)$ elements chosen uniformly and independently at random. Let $\AVG$ be the average of all elements and $\AVG'$ be the $average$ of sampled elements. Then with high constant probability, we have: $\AVG'-\epsilon\leq \AVG \leq \AVG' +\epsilon.$
\end{lemma}

\begin{theorem}
\label{partition}
\textsc{Approx-Partition-Function}(${\lambda,\epsilon}$) is an $\epsilon n$-approximation algorithm for $\log Z(G,\lambda)$ with query complexity
$\tilde{O}\left((1/\epsilon)^{\tilde{O}(\sqrt{\Delta})}\right)$.
\end{theorem}

\begin{proof}
Let $A=\sum_{v\in V}-\log(\textsc{Approx-Marginal$^*$}(\lambda,\epsilon/2,v))$. By Proposition~\ref{marginal} and Equation (\ref{z}), $A$ is an $\epsilon n/2$-approximation of $\log Z(G,\lambda)$.
By Lemma~\ref{sampling}, there exists some constant $C$ such that approximating the marginal probability at $\lceil C/\epsilon^2\rceil$ sampled nodes gives an $\epsilon n/2$-approximation of $A$ with high probability.
This implies an $\epsilon n$-approximation of $\log Z(G,\lambda)$ with high probability.
The query complexity of this algorithm is $\lceil C/\epsilon^2\rceil\cdot\overline{\mathcal{Q}}(\epsilon/2,\Delta)=\tilde{O}\left((1/\epsilon)^{\tilde{O}(\sqrt{\Delta})}\right)$.
\end{proof}

Note that the size of any maximal matching is always lower bounded by
$\frac{m}{2\Delta-1}$, where $m$ is the number of edges. In
particular, since $Z(G,1)$ is the total number of matchings, we have
$\frac{m}{2\Delta-1}\log 2\leq \log Z(G,1)\leq m\log 2\leq
\frac{n\Delta}{2}\log 2$ so that if $m=\Omega(n)$,
we also have $\log Z(G,1)=\Theta(n)$. Hence, if $\epsilon$ and
$\Delta$ are constants and $m=\Omega(n)$, the error in the output of
our algorithm is of the same order as the evaluated quantity. This is
in contrast with the FPTAS (Fully Polynomial-Time Approximation Scheme) in \cite{bayati2007simple} or the FPRAS (Fully Polynomial-time Randomized Approximation Scheme) in \cite{jerbook,sinbook}  which outputs an $\epsilon$-approximation instead of an $\epsilon n$-approximation.
Of course, we can let $\epsilon$ tend to $0$ with $n$ like $c/n$ in
Theorem~\ref{partition}, so that our result (when $\Delta$ is
constant) is consistent with the FPTAS in
\cite{bayati2007simple}.
Indeed, in this case, clearly no sampling is required and if we replace the sampling step by a visit of each vertex, our algorithm is the same as in \cite{bayati2007simple}.

When we assume $\Delta$ to be fixed, the query complexity of the above algorithm is polynomial in $1/\epsilon$. Next we give a lower bound on the query complexity which is quadratic in $1/\epsilon$. In the proof, we use a lower bound result from \cite{chazelle2005approximating}, which is based on Yao's Minimax Principle~\cite{yao1977probabilistic}. See \apref{proof_lowerboundZ} for the proof of the following theorem.

\begin{theorem}
\label{lowerboundZ}
Any deterministic or probabilistic $\epsilon n$-approximation algorithm for $\log
Z(G,\lambda)$ needs $\Omega(1/\epsilon^2)$ queries on average. It is
assumed that $\epsilon>C/\sqrt{n}$ for some constant $C$.
\end{theorem}

\section{Approximating Matching Statistics}
\label{sec:matching statistics}
We define the average size $E(G,\lambda)$ and the entropy $S(G,\lambda)$ of a matching as:
\[ E(G,\lambda)= \sum_{M\in \mathbb{M}}|M|\;\pi_{G,\lambda}(M) \quad\mbox{ and }\quad S(G,\lambda)=-\sum_{M\in \mathbb{M}}\pi_{G,\lambda}(M)\;\log \pi_{G,\lambda}(M). \]

The following algorithm estimates $E(G,\lambda)$, where $C$ is a fixed constant. \\

\begin{algorithm}{Approx-Matching-Statistics}{\lambda,\epsilon}
s\=\lceil C/\epsilon^2\rceil\\
U\=\text{random multi-subset of $V$ with $s$ elements}\\
\RETURN  n-(n/2s) \cdot\sum_{u\in U}(\CALL{Approx-Marginal}(\lambda,\epsilon/2,u))
\end{algorithm}

\begin{theorem}
\label{statistics}
\textsc{Approx-Matching-Statistics}$(\lambda,\epsilon)$ is an $\epsilon n$-approximation algorithm for $E(G,\lambda)$ with query complexity $\tilde{O}\left((1/\epsilon)^{\tilde{O}(\sqrt{\Delta})}\right)$.
In addition, any $\epsilon n$-approximation algorithm for $E(G,\lambda)$ needs $\Omega(1/\epsilon^2)$ queries.
\end{theorem}

\begin{proof}
Let $A=\sum_{v\in V} \textsc{Approx-Marginal}(\lambda,\epsilon/2,v)$. By Proposition~\ref{marginal}, $A$ is an $\epsilon n/2$-approximation of $\sum_{v\in V}p_{G,\lambda}(v)$.
By Lemma~\ref{sampling}, there exists some constant $C$ such that approximating the marginal probability at $\lceil C/\epsilon^2\rceil$ sampled nodes gives an $\epsilon n/2$-approximation of $A$ with high probability.
This implies an $\epsilon n$-approximation of $\sum_{v\in V}p_{G,\lambda}(v)$ with high probability.
Since $E(G,\lambda)=n-\sum_{v\in V}p_{G,\lambda}(v)/2$, we thus get an $\epsilon n$-approximation of $E(G,\lambda)$ with high probability. The query complexity of this algorithm is $\lceil C/\epsilon^2\rceil\cdot\overline{\mathcal{Q}}(\epsilon/2,\Delta)=\tilde{O}\left((1/\epsilon)^{\tilde{O}(\sqrt{\Delta})}\right)$.
The query lower bound is obtained similarly as Theorem~\ref{lowerboundZ}.
\end{proof}

\begin{corollary}\label{cor:stat}
We have an $\epsilon n$-approximation algorithm for
$S(G,\lambda)$ with query complexity $\tilde{O}\left((1/\epsilon)^{\tilde{O}(\sqrt{\Delta})}\right)$.
In addition, any $\epsilon n$-approximation algorithm for $S(G,\lambda)$ needs $\Omega(1/\epsilon^2)$ queries.
\end{corollary}
\begin{proof}
A simple calculation gives: $S(G,\lambda) =\log Z(G,\lambda) - \log \lambda \cdot E(G,\lambda).$
Let $\widehat{Z}$ be the output of $\textsc{Approx-Partition-Function}(\lambda,\epsilon/2)$ and $\widehat{E}$ be the output of $\textsc{Approx-Matching-Statistics}(\lambda,\epsilon/(2\log\lambda))$.
By Theorem~\ref{partition} and Theorem~\ref{statistics}, $\widehat{Z}-\log\lambda \cdot \widehat{E}$ is an $\epsilon n$-estimate of $S(G,\lambda)$ with high probability. Both $\widehat{Z}$ and $\widehat{E}$ are computed using $\tilde{O}\left((1/\epsilon)^{\tilde{O}(\sqrt{\Delta})}\right)$ queries.
The query lower bound is obtained similarly as Theorem~\ref{lowerboundZ}.
\end{proof}

\section{Some Applications}
\label{sec:applications}
So far, we did consider that the parameter $\lambda$ is fixed.
Letting $\lambda$ grow with $\frac{1}{\epsilon}$ allows us to get new results for the permanent of a matrix.
There is a FPRAS for the permanent of a matrix with non-negative entries~\cite{jerrum2004polynomial}.
When the matrix is the adjacency matrix of a \emph{constant degree
  expander} graph, there is a PTAS to estimate the permanent within a
multiplicative factor
$(1+\epsilon)^n$~\cite{gamarnik2010deterministic}. Using a key
technical result of \cite{gamarnik2010deterministic}, we get a sublinear-time algorithm within the same multiplicative factor (see \apref{sec:perm}).

Note that for a fixed graph $G$, if $\lambda\to \infty$ then the
distribution $\pi_{G,\lambda}$ converges toward the uniform
distribution on maximum matchings. Indeed, using a bound derived in
\cite{bordenave2012matchings}, we can show
that if $\lambda$ grows exponentially with $\frac{1}{\epsilon}$ our
technique allows to approximate the size of a maximum
matching (see \apref{sec:MM}). However our algorithm performs badly with respect to
\cite{yoshida2009improved}.

\subsubsection*{Acknowledgments.}
The authors acknowledge the support of the French \emph{Agence Nationale de la Recherche (ANR)} under reference ANR-11-JS02-005-01 (GAP project).

\bibliography{references}

\ifnum\fullversion=1

\appendix
\clearpage

\section{Proof of Proposition~\ref{lowerboundP}}
\label{proof_lowerboundP}
The idea is to define a certain distribution of graphs and show that on that random
input, any deterministic algorithm for approximating $p_{G,\lambda}(v)$ within an additive error $\epsilon$ requires $\underline{\mathcal{Q}}(\epsilon,\Delta)$ queries.
By Yao’s Minimax Principle~\cite{yao1977probabilistic}, the result follows.

We start by studying the full $\Delta$-ary tree of height $h$, let it be $T^h$. Define  $y_1=1$ and $y_k = \left( 1+\lambda\Delta y_{k-1}\right)^{-1}\mbox{ for }k\geq 2.$
Clearly $y_h$ is equal to the value computed at the root of
$T^h$ by the recursion (\ref{eq:recx}).

\begin{lemma}\label{lem:x}
We have $\lim_{k\to \infty} y_k= \frac{2}{1+\sqrt{1+4\lambda\Delta}}$.
Define $\underline{h}(\epsilon,\Delta) = \sup\{h,\: |y_{h}-y_{h-1}|\geq \epsilon
\}$. Then we have $\displaystyle\lim_{\Delta\to \infty}\frac{1}{\sqrt{\Delta}}\lim_{\epsilon\to
  0}\frac{\underline{h}(\epsilon,\Delta)}{\log(1/\epsilon)} = \sqrt{\lambda}$.
\end{lemma}
\begin{proof}
To study $y_k$, we introduce the auxiliary sequence:
$f_k=f_{k-1}+\lambda \Delta f_{k-2}$ for $k\geq 3$ and $f_1=f_2=1$ so that
$y_k=f_{k}/f_{k+1}$ for $k\geq 1$. Let $\alpha =(1+\sqrt{1+4\lambda\Delta})/2$ and
$\beta=(1-\sqrt{1+4\lambda\Delta})/2$, we have
$f_k=\frac{1}{2\alpha-1}(\alpha^k-\beta^k)$ and the first statement of the lemma follows.

A simple computation gives:
$|y_{k}-y_{k-1}| \sim
C_\Delta(D_\Delta)^{k-1}$ as $k\to \infty$, where $C_\Delta=\frac{\alpha^2+\beta^2-2\alpha\beta}{\alpha^3}$ and $D_\Delta=\frac{|\beta|}{\alpha}$.
So $\displaystyle \lim_{\epsilon\to 0}\frac{\underline{h}(\epsilon,\Delta)}{\log(1/\epsilon)} = -\frac{1}{\log D_\Delta}$.
Since $\log D_\Delta \sim -\frac{1}{\sqrt{\lambda \Delta}}$ as $\Delta\to \infty$,
the second statement of the lemma follows.
\end{proof}

Denote $L$ to be the set of leaves of the tree $T^h$. Let ${\bf e} = (e_u,\:u\in L)$ be a vector
of $\{0,1\}^L$. Define $T^h({\bf e})$ to be the tree
obtained from $T^h$ in the following way: every non-leaf node in $T^h$
remains in $T^h({\bf e})$, and a leaf $u$ in $T^h$ remains in
$T^h({\bf e})$ iff $e_u=1$.
 (As a result, we see that the recursion (\ref{eq:recx}) is valid with
 $x_u(v) = e_u$ for all $u\in L$).
Define the vector ${\bf x}^h({\bf e})=(x^h_u({\bf e}),\:u\in T^h({\bf e}))$ as defined in (\ref{eq:recx}).
We denote by $x^h({\bf e})$ the value of the component of ${\bf
  x}^h({\bf e})$ corresponding to the root of the tree.
For any node $u$ of depth $k<h$ in the tree $T^h$, simple monotonicity arguments show that for $h-k$ even, $y_{h-k}=x_u^h({\bf 0})\leq
x_u^h({\bf e})\leq x_u^h({\bf 1})=y_{h-k+1}$, and for $h-k$ odd, $y_{h-k+1}=x_u^h({\bf 1})\leq
x_u^h({\bf e})\leq x_u^h({\bf 0})=y_{h-k}$.
Define $d^h_u({\bf e}) =|x^h_u({\bf 1})-x^h_u({\bf e})|$ for all $u\in T^h$.
We have:
\BEAS
d^h_u({\bf e}) &=&\left| \frac{1}{1+\lambda \sum_{w\in\child{u}} x^h_w({\bf
      1})}-\frac{1}{1+\lambda \sum_{w\in \child{u}} x^h_w({\bf e})}\right|\\
&\leq& \frac{1}{1+\lambda
    \sum_{w\in \child{u}} x^h_w({\bf 1})}\cdot\frac{1}{ 1+\lambda
    \sum_{w\in \child{u}} x^h_w({\bf e})}\cdot\lambda \sum_{w\in \child{u}}d^h_w({\bf e})
\EEAS
Notice that the first term is  $y_{h-k+1}$ by definition.
For every $w\in\child{u}$, $x^h_w({\bf e})\geq \min(y_{h-k-1},y_{h-k})$, so the second term  $\leq \max(y_{h-k},y_{h-k+1})$.
Thus we have:
\begin{equation}
\label{recursion}
d^h_u({\bf e})\leq y_{h-k+1}\cdot\max(y_{h-k},y_{h-k+1}) \cdot\lambda\sum_{w\in \child{u}}d^h_w({\bf e}).
\end{equation}
Consider $u$ to be the root of the tree. Then $d^h_u({\bf e})=|x^h({\bf e})-x^h({\bf 1})|$.
Iterating $h-1$ times the inequality~\eqref{recursion}, we have:
\[
|x^h({\bf e})-x^h({\bf 1})| \leq \sum_{u\in L}e_u \cdot\prod_{i=2}^{h} \lambda y_i \cdot \max(y_{i-1},y_i).
\]
Using the fact that $|y_{i-1}-y_i|/y_i$ decreases exponentially, it is easy to show that there exists a constant $C$, such that
$\prod_{i=2}^{h}\max(y_{i-1},y_{i})<C\cdot \prod_{i=2}^{h}y_{i}$, for all $h$.
So we have:
\begin{equation}\label{accumulate}|x^h({\bf e})-x^h({\bf 1})| \leq \sum_{u\in L}e_u \cdot C\prod_{i=2}^{h} \lambda y_i^2=\sum_{u\in L}e_u\cdot\frac{C\lambda^{h-1}}{ f_{h+1}^2}= \sum_{u\in L}e_u \cdot O\left((\lambda/\alpha^2)^h\right).\end{equation}

Let $h=\underline{h}(2\epsilon,\Delta)$ and consider the uniform distribution on the set of trees $T^h({\bf e})$, for every ${\bf e}\in\{0,1\}^L$.
Suppose $\mathcal{A}$ is an deterministic algorithm for approximating $p_{G,\lambda}(v)$  using $o\left(\Delta^h\right)$ queries under this distribution.
Then $\mathcal{A}$ has visited at most $M=o(\Delta^h)$ positions in the $h^{\rm th}$ level before stopping.
Let ${\bf e^-}$ (resp. ${\bf e^+}$) be a vector in $\{0,1\}^{L}$, where the $M$ positions visited by $\mathcal{A}$ have fixed values and all other positions are $0$ (resp. $1$). We will show that $\mathcal{A}$ cannot provide an $\epsilon$-approximation of the marginal for both trees $T^h({\bf e^-})$ and $T^h({\bf e^+})$. By Equation (\ref{accumulate}), $|x^h({\bf e^+})-x^h({\bf 1})|=M\cdot O\left((\lambda/\alpha^2)^h\right)=o\left((\lambda\Delta/\alpha^2)^h\right)=o(|\beta/\alpha|^h)=o(\epsilon).$
Similarly we can show that $|x^h({\bf e^-})-x^h({\bf 0})|=o(\epsilon).$ By the definition of $h$, $|x^h({\bf 1})-x^h({\bf 0})|\geq 2\epsilon$, so we have $|x^h({\bf e^-})-x^h({\bf e^+})|> \epsilon$. Since $\mathcal{A}$ is deterministic, it cannot provide an estimate which is an $\epsilon$-approximation of $x^h({\bf e^-})$ and an $\epsilon$-approximation of $x^h({\bf e^+})$.
Thus we obtain a lower bound
$\Omega(\underline{\mathcal{Q}}(\epsilon,\Delta))$ for approximating the marginal within an error $\epsilon$, where $\underline{\mathcal{Q}}(\epsilon,\Delta)=\Delta^{\underline{h}(2\epsilon,\Delta)}$  satisfies Equation~\eqref{eq:limT}.

Notice that the trees studied above have maximal degree $\Delta+1$
but changing $\Delta$ to $\Delta+1$ will not affect the statement of
the proposition. Note also that Kahn and Kim did study in
\cite{kahn1998random} similar recursions for matchings in regular
graphs but with a different perturbation at each level of the tree.

\section{Proof of Theorem~\ref{lowerboundZ}}
\label{proof_lowerboundZ}
For $s\in\{0,1\}$, let $\mathcal{D}_s$ denote the distribution induced by
setting a binary random variable to $1$ with probability
$p_s=(1+(-1)^s\epsilon)/2$ (and $0$ else). We define a distribution
$\mathcal{D}$ on $m$-bit strings as follows: (1) pick $s=1$ with
probability $1/2$; (2) draw a random string from $\{0,1\}^m$ by
choosing each bit $b_i$ from $\mathcal{D}_s$ independently. The following lemma is proved in  \cite{chazelle2005approximating}.

\begin{lemma}
\label{probabilistic_lower_bound}
 Any probabilistic algorithm that can guess the value of $s$ with a probability of error below $1/4$ requires $\Omega(1/\epsilon^2)$ bit lookups on average.
\end{lemma}

In order to get the lower bound query complexity of $\log Z(G,\lambda)$, the idea is to create an \mbox{$n$-node} random graphs $G_s$ depending on $s\in\{0,1\}$ such that $\log Z(G_0, \lambda)-\log Z(G_1,\lambda)> \rho\epsilon n$ for some constant $\rho$ with high probability. So if there exists a $(\rho\epsilon n/3)$-approximation algorithm for $\log Z(G,\lambda)$ using $o(1/\epsilon^2)$ queries, then we can differentiate $G_0$ and $G_1$ thus obtain the value of $s$ with high probability using also $o(1/\epsilon^2)$ queries, which contradicts with the lower bound complexity in Lemma~\ref{probabilistic_lower_bound}.

Consider the graph $G$ consisting of $n$ isolated vertices $v_1,\cdots, v_n$. Pick $s\in\{0,1\}$ uniformly at random and take a random $\lfloor n/2 \rfloor$-bit string $b_1,\cdots,b_{\lfloor n/2 \rfloor}$ with bits drawn from $\mathcal{D}_s$ independently. Next, add an edge between $v_{2i-1}$ and $v_{2i}$ if and only if $b_i=1$. Notice that the function $\log Z(G,\lambda)$ is additive over disjoint components, so $\log Z(G_s,\lambda)=\sum_{i=1}^{\lfloor n/2\rfloor} x_i$, where $\{x_i\}_{1\leq i\leq n}$ are independent random variables, and each $x_i$ equals  $\log (1+\lambda)$ with probability $(1+(-1)^s\epsilon)/2$ and equals $0$ otherwise. For any two graphs $G_0$ and $G_1$ derived from $\mathcal{D}_0$ and $\mathcal{D}_1$ respectively, we have $\mathbb{E}[\log Z(G_0,
\lambda)]-\mathbb{E}[\log Z(G_1,\lambda)]=\log(1+\lambda)\cdot\epsilon\lfloor n/2\rfloor$. When
$\epsilon>C/\sqrt{n}$ for some constant $C$ large enough, we have $|\mathbb{E}[\log
Z(G_0, \lambda)]-\log Z(G_0, \lambda)|<\log(1+\lambda)\cdot\epsilon n/10$ and
$|\mathbb{E}[\log Z(G_1, \lambda)]-\log Z(G_1, \lambda)|<\log(1+\lambda)\cdot\epsilon n/10$
with high probability. Thus $\log Z(G_0, \lambda)-\log
Z(G_1,\lambda)> \log(1+\lambda)\cdot\epsilon n/5$ with high probability. Together with Lemma~\ref{probabilistic_lower_bound}, we know that any
probabilistic $\left(\log(1+\lambda)\cdot\epsilon n/15\right)$-approximation algorithm for $\log Z(G,\lambda)$ needs
$\Omega(1/\epsilon^2)$ queries on average, thus the statement of the theorem follows.

\section{Permanent of expander graphs}\label{sec:perm}
Consider a bi-partite graph $G$ with the node set $V=X \cup Y$, where $|X|=|Y|=n$. For every $S\subset V$, denote $N(S)$ to be the set of nodes adjacent to at least one node in $S$. For $\alpha>0$, a graph is an \emph{$\alpha$-expander} if for every subset $S\subset X$ and every subset $S\subset Y$, as soon as $|S|\leq n/2$, we have $|N(S)|\geq(1+\alpha)|A|$. Let $A=(a_{i,j})$ be the corresponding adjacency matrix of $G$, i.e., the rows and columns of $A$ are indexed by nodes of $X$ and $Y$ respectively, and $a_{i,j}=1$ iff $(x_i,y_j)$ is an edge in $G$. Let $\PERM$ denote the permanent of $A$. We already know that computing the permanent of a matrix is \#P-complete, even when the entries are limited to 0 and 1, so we look for an estimate of $\PERM$.

In~\cite{gamarnik2010deterministic}, Gamarnik and Katz provided a polynomial-time approximation algorithm for $\PERM$ with multiplication factor $(1+\epsilon)^n$. A key tool in their analysis is the following lemma.

\begin{lemma}
\label{permanent}
(\cite{gamarnik2010deterministic})
Let $G$ be a bi-partite $\alpha$-expander graph of bounded degree $\Delta$. Then for every $\lambda>0$, we have:
$$1\leq\frac{Z(G,\lambda)}{\lambda^n \PERM}\leq e^{O(n\lambda^{-1}\log^{-1}(1+\alpha)\log\Delta)}.$$
\end{lemma}

Based on this lemma and combined with our result in Section~\ref{sec:partition function}, we then obtain an algorithm to estimate $\log\PERM$ within an error $\epsilon n$ in sublinear time. This improves the polynomial-time algorithm in~\cite{gamarnik2010deterministic} which has an equivalent approximation factor.

\begin{proposition}
Let $G$ be a bi-partite $\alpha$-expander graph of bounded degree $\Delta$. Then there is an $\epsilon n$-approximation algorithm for $\log \PERM$ with query complexity $\tilde{O}\left((1/\epsilon)^{\tilde{O}(\sqrt{\Delta/(\epsilon\alpha)}}\right)$.
\end{proposition}

\begin{proof}
Take $\lambda=\Theta\left(\log\Delta/(\epsilon\alpha)\right)$ so that $O(\lambda^{-1}\log^{-1}(1+\alpha)\log\Delta)<\epsilon/2$.
The algorithm \textsc{Approx-Partition-Function}$(\lambda,\epsilon/2)$ in Section~\ref{sec:partition function} provides an estimate $\widehat{Z}$ of $\log Z(G,\lambda)$ using $\tilde{O}\left((1/\epsilon)^{\tilde{O}(\sqrt{\Delta/(\epsilon\alpha)}}\right)$ queries, and with high constant probability, $\widehat{Z}$ is an $(\epsilon n/2)$-approximation of $\log Z(G,\lambda)$.
From Lemma~\ref{permanent}, $\log\widehat{Z}-n\log\lambda$ is an $\epsilon n$-approximation of $\log\PERM$ with high constant probability.
\end{proof}

\section{Maximum matching size}\label{sec:MM}

As $\lambda\to\infty$, $E(G,\lambda)$ tends to the size of a maximum matching, let it be $\OPT$.

\begin{lemma}
\label{activity}
For any $\epsilon>0$, taking $\lambda=e^{\frac{\Delta\log 2}{2\epsilon}}$ in $E(G,\lambda)$ leads to an $\epsilon n$-approximation of the maximum matching size.
\end{lemma}

\begin{proof}
Since $m\leq \frac{\Delta n}{2}$ in a bounded degree graph, this lemma follows directly from Lemma~12 in \cite{bordenave2012matchings}, which proves that $E(G,\lambda)\leq \OPT\leq E(G,\lambda)+\frac{m\log2}{\log \lambda}$.
\end{proof}
Applying the above value of $\lambda$ to the \textsc{Approx-Matching-Statistics}$(\lambda,\epsilon)$ algorithm of Section~\ref{sec:matching statistics} leads to the following proposition.
\begin{proposition}
We have an $\epsilon n$-approximation algorithm for the maximum matching size with query complexity $\tilde{O}\left((1/\epsilon)^{e^{\tilde{O}(\Delta/\epsilon)}}\right)$.
\end{proposition}

\begin{remark*}
The query complexity of our algorithm is double exponential in $\Delta/\epsilon$, which is outperformed by \cite{yoshida2009improved}.
\end{remark*}

\section{Independent Sets}\label{sec:other}
Let $\mathbb{I}$ be the set of independent sets of $G$. The partition function of the system is defined by $Z_I(G,\lambda) = \sum_{I\in \mathbb{I}} \lambda^{|I|}$, and the Gibbs distribution on the space $\mathbb{I}$ is defined by $\pi_{G,\lambda} (I) =\frac{\lambda^{|I|}}{Z_I(G,\lambda)}$. For every  $v\in V$, define $p_{G,\lambda}(v) := \pi_{G,\lambda}(v\notin I) =
\sum_{I \not\ni v}\pi_{G,\lambda} (I)$, where $I \not\ni v$ is an independent set not containing $v$.

Notice that $1/Z_I(G,\lambda)$ is exactly the probability of the empty
set in Gibbs distribution, which is also equal to $\prod_{v\in V}p_{G_v,\lambda}(v)$, where $G_v$ is defined similarly as in Section~\ref{sec:partition function}.
Hence we have: \[\log Z_I(G,\lambda)=\sum_{v\in V}-\log p_{G_v,\lambda}(v).\]
However it is well-known that the correlation decay implying a result
similar to Lemma \ref{lem:bayati} does not hold for all values of
$\lambda$ in the independent set problem.
Indeed Weitz in \cite{weitz2006counting} gave a FPTAS for estimating
$Z(G,\lambda)$ up to the critical activity $\lambda_c$ for the uniqueness of the
Gibbs measure on the infinite $\Delta$-regular tree.
We adapt his approach and get the analogue of Lemma \ref{lem:bayati}.

\begin{lemma}(Corollary~2.6 in~\cite{weitz2006counting})
\label{strong spatial mixing}
Let $\lambda$ be such that $0<\lambda<\lambda_c(\Delta)=\frac{\Delta^\Delta}{(\Delta-1)^{\Delta+1}}$.
There exists some decaying rate $\delta$ with $\delta(l)=O(e^{-\alpha l})$ for some $\alpha>0$, such that every graph of maximum degree $\Delta+1$ with activity $\lambda$ exhibits \emph{strong spatial mixing (see \cite{weitz2006counting})} with rate $\delta$.
\end{lemma}

A similar approach as in Section \ref{sec:partition function} leads to the following result.

\begin{proposition}
Let $0<\lambda<\lambda_c(\Delta)$. We have an $\epsilon n$-approximation algorithm for $\log
Z_I(G,\lambda)$ with query complexity polynomial in $\frac{1}{\epsilon}$ for any graph $G$ with maximum
degree $\Delta+1$. In addition, any $\epsilon n$-approximation algorithm for $\log Z_I(G,\lambda)$ needs $\Omega(1/\epsilon^2)$ queries.
\end{proposition}

\begin{proof}
The algorithm is the following:
\begin{enumerate}
\item take a sampling of $\Theta(1/\epsilon^2)$ vertices;
\item estimate $\log p_{G_v,\lambda}(v)$ with an additive error $\epsilon/2$ for every sampled vertex $v$;
\item compute an estimate of $\log Z_I(G,\lambda)$ from these estimates.
\end{enumerate}
Steps 1 and 3 are the same as in Section~\ref{sec:partition function}, so we focus on Step 2, i.e., to find two bounds $p_0$ and $p_1$ such that $p_0\leq p_{G_v,\lambda}(v)\leq p_1$ and $\log p_1-\log p_0 \leq \epsilon/2$.
We will do this by exploring a constant-size path-tree (also called self-avoiding-walk in \cite{weitz2006counting}) rooted at $v$.

By Lemma~\ref{strong spatial mixing}, there exists some constants $\alpha>0$ such that $\delta(l)=O(e^{-\alpha l})$ is the strong spatial mixing rate.
Let $\epsilon'=\frac{\epsilon}{2(1+\lambda)}$.
Take $h=(\log\frac{1}{\epsilon'})/\alpha+O(1)$ so that $\delta(h)=\epsilon'$.
For the path-tree rooted at $v$ and truncated at level $h$ (resp. level $h+1$), we compute the marginal $p_0$ (resp. $p_1$) at the root $v$ recursively by Equation~(1) in \cite{weitz2006counting}, which takes $O(\Delta^h)$ (resp. $O(\Delta^{h+1})$) queries.
Without loss of generality, we assume that $p_0\leq p_1$.
By the monotonicity argument as for Inequality~\eqref{alternate}, we have $p_0\leq p_{G_v,\lambda}(v)\leq p_1$.
By the property of strong special mixing, $p_1-p_0\leq\delta(h)=\epsilon'$. So $\log p_1-\log p_0\leq \log(1+\frac{\epsilon'}{p_0})\leq \log(1+(1+\lambda)\epsilon')\leq \epsilon/2$, where the second inequality is because $p_0\geq \frac{1}{1+\lambda}$.

The number of queries in our computation is $O(\Delta^{h+1}/\epsilon^2)$, which is polynomial in $1/\epsilon$.
The query lower bound is obtained similarly as Theorem~\ref{lowerboundZ}.
\end{proof}

\begin{remark*}
If $\Delta< 5$, we have $\lambda_c(\Delta)>1$.
Then we can approximate $Z_I(G,1)$, which is the number of independent sets in $G$.
\end{remark*}

\section{Tests on large graphs}\label{sec:tests}

We show the performance of our algorithm on the average size of a matching $E(G,1)$ on large real-world graphs from Stanford large network dataset collection\footnote{see http://snap.stanford.edu/data/index.html}. Our algorithm performs well on both small degree graphs and small average-degree graphs.
The tests are based on:
\begin{itemize}
\item microprocessor: intel core i5 750 (2.67 GHz,  256KB L2/core, 8MB L3)
\item memory: RAM 4 Go
\item compiler: g++ version 4.4.3, option -O2
\item operating system: Linux Ubuntu 10.04
\end{itemize}

\subsection{Small degree graphs}
Consider the three road network graphs from Stanford large network dataset collection, where $\Delta$ is at most $12$. Intersections and endpoints are represented by nodes and the roads connecting these intersections or road endpoints are represented by undirected edges.
\begin{itemize}
\item roadNet-CA : California road network with $n=1965206$, $\Delta=12$

\item roadNet-PA : Pennsylvania road network with $n=1088092$, $\Delta=9$

\item roadNet-TX : Texas road network with $n=1379917$, $\Delta=12$
\end{itemize}
We test the \textsc{Approx-Matching-Statistics} algorithm in Section~\ref{sec:matching statistics} for increasing values of $1/\epsilon$.  The following figure gives the execution time of our program with respect to $1/\epsilon$, where the three curves correspond to the three graphs above.
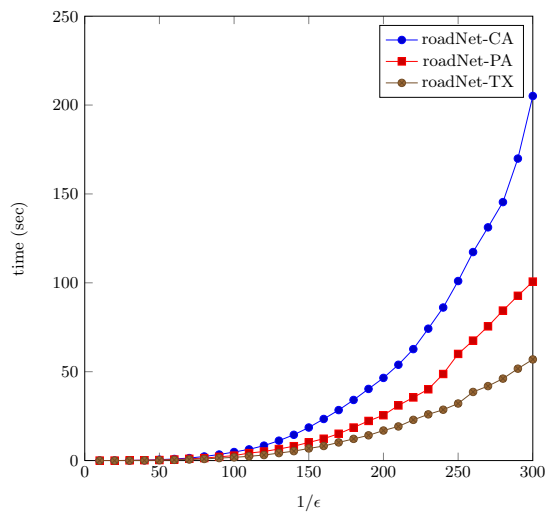
\begin{figure}
\centering
\scalebox{0.7}{
\begin{tikzpicture}
\begin{axis}[
height=10cm,
width=10cm,
xlabel=$1/\epsilon$,
ylabel=time (sec),
xmin=0,
xmax=300,
ymin=0,
ymax=250,
ytick pos=left
]
\addplot coordinates { 
(10 , 0.0)
(20 , 0.0)
(30 , 0.1)
(40 , 0.2)
(50 , 0.5)
(60 , 0.9)
(70 , 1.4)
(80 , 2.3)
(90 , 3.4)
(100 , 4.8)
(110 , 6.3)
(120 , 8.4)
(130 , 11.2)
(140 , 14.5)
(150 , 18.6)
(160 , 23.4)
(170 , 28.4)
(180 , 34.1)
(190 , 40.3)
(200 , 46.5)
(210 , 53.9)
(220 , 62.7)
(230 , 74.2)
(240 , 86.1)
(250 , 101.0)
(260 , 117.3)
(270 , 131.2)
(280 , 145.4)
(290 , 169.9)
(300	, 205.1)
};
\addlegendentry{roadNet-CA}

\addplot coordinates { 
(10 , 0.0)
(20 , 0.0)
(30 , 0.1)
(40 , 0.2)
(50 , 0.3)
(60 , 0.6)
(70 , 1.0)
(80 , 1.4)
(90 , 2.1)
(100 , 2.8)
(110 , 4.2)
(120 , 5.1)
(130 , 6.5)
(140 , 8.1)
(150 , 10.2)
(160 , 12.3)
(170 , 15.1)
(180 , 18.6)
(190 , 22.3)
(200 , 25.5)
(210 , 31.0)
(220 , 35.5)
(230 , 40.1)
(240 , 48.7)
(250 , 60.0)
(260 , 67.5)
(270 , 75.5)
(280 , 84.4)
(290 , 92.7)
(300 , 100.7)
};
\addlegendentry{roadNet-PA}

\addplot coordinates { 
(10 , 0.0)
(20 , 0.0)
(30 , 0.0)
(40 , 0.1)
(50 , 0.2)
(60 , 0.4)
(70 , 0.6)
(80 , 0.9)
(90 , 1.4)
(100 , 1.8)
(110 , 2.4)
(120 , 3.1)
(130 , 4.2)
(140 , 5.3)
(150 , 6.8)
(160 , 8.2)
(170 , 10.2)
(180 , 12.2)
(190 , 14.2)
(200 , 16.9)
(210 , 19.2)
(220 , 22.9)
(230 , 26.0)
(240 , 28.6)
(250 , 32.1)
(260 , 38.6)
(270 , 41.9)
(280 , 46.1)
(290 , 51.7)
(300 , 56.9)
};
\addlegendentry{roadNet-TX}
\end{axis}
\end{tikzpicture}
}
\end{figure}

\subsection{Small average-degree graphs}
Now we extend our algorithm to graphs with small average degree, since these graphs are of much practical interest in the real world.
For a large degree node $v$ in such a graph, it takes too many queries for the $\textsc{Approx-Marginal}$ algorithm to output an $\epsilon$-approximation of marginal distribution.
The solution is to compute a less accurate approximation of the marginal distribution for large degree nodes.
Notice that in a graph with small average degree, the number of large degree nodes is limited.
So the approximation performance in general should be good, as shown in the tests later.

In the algorithm, we take a constant number (e.g., 400 in our program) of sampled nodes, and for every node $v$ in the sampling, we limit the time (and thus the number of queries) for estimating $p_{G,1}(v)$ and get an estimate with error possibly larger than $\epsilon$.
More precisely, we calculate $x_v^1(v), x_v^2(v),\cdots$, and stop if the accumulated time of these calculations exceeds one second.
Let $x_v^k(v)$ be the last estimate of $p_{G,1}(v)$ computed by the algorithm.
By Inequality~\eqref{alternate}, the exact value of $p_{G,1}(v)$ is between $x_v^{k-1}(v)$ and $x_v^{k}(v)$.
Thus $x_v^k(v)$ is an estimate of $p_{G,1}(v)$ with error at most $|x_v^k(v)-x_v^{k-1}(v)|$. Let this error be $\epsilon_v$.
Define $\bar{\epsilon}$ to be the average value of $\epsilon_v$ for all $v$ in the sampling.
From the estimates of $p_{G,1}(v)$ on sampled nodes, we get an estimate of $E(G,1)$, which is an $(\bar{\epsilon}+0.05)$-approximation of $E(G,1)$ with high probability.

Below is the performance of our algorithm on
\emph{Brightkite-edges}\footnote{This was once a location-based social networking service provider where users shared their locations by checking-in; the friendship network was collected using their public API. In this graph, $n=58228$, average degree=3.7, $\Delta=1134$.} and
\emph{CA-CondMat}\footnote{This is a collaboration network of Arxiv Condensed Matter category; there is an edge if authors coauthored at least one paper. In this graph, $n=23133$, average degree=8.1, $\Delta=280$.}
from Stanford large network dataset collection.

\begin{itemize}
\item \emph{Brightkite-edges}:
    We get an estimate 32626 of $E(G,1)$ in 12.9 seconds with $\bar{\epsilon}=0.0016$.
\item \emph{CA-CondMat}:
    We get an estimate 17952 of $E(G,1)$ in 287.5 seconds with $\bar{\epsilon}=0.0291$.
\end{itemize}

\fi

\end{document}